\documentclass[11pt]{article}

\usepackage{fullpage}
\usepackage{times}
\usepackage{amsfonts}
\usepackage{amssymb}
\usepackage{amsthm}
\usepackage{latexsym}
\usepackage{amsmath}

\newtheorem{thm}{Theorem}

\newtheorem{lemma}[thm]{Lemma}

\def \OPT  {\mbox{\rm OPT}}

\begin{document}
\title{The Geometry of Scheduling}
\author{
Nikhil Bansal
\thanks{IBM T. J. Watson Research Center,
Yorktown Heights, NY 10598 USA.
E-mail: nikhil@us.ibm.com}
\and
 Kirk Pruhs 
 \thanks{
Computer Science Department, 
University of Pittsburgh, 
Pittsburgh, PA 15260 USA.
Email: kirk@cs.pitt.edu.
Supported in part by
NSF grants CNS-0325353,  IIS-0534531, and CCF-0830558, and
an IBM Faculty Award.} 
} 
\date{}
\maketitle

\begin{abstract}
We consider the following general scheduling problem: 
The input consists of $n$ jobs, each with an arbitrary release time,
size, and a monotone function 
specifying the cost incurred when the job is completed at a particular time.
The objective is to find a preemptive schedule of minimum aggregate cost.
This problem formulation is general enough to include many natural scheduling 
objectives, such as weighted flow, weighted tardiness, and sum of flow squared.

The main contribution of this paper 
is a randomized polynomial-time algorithm with an approximation ratio 
$O(\log \log nP )$,
where $P$ is the maximum job size.
We also give an $O(1)$ approximation in the special case when all jobs have identical release times.
Initially, we show how to reduce this scheduling problem
to a particular geometric set-cover problem.
We then consider a natural linear programming formulation of this geometric set-cover problem, 
strengthened by adding knapsack cover inequalities,
and show that rounding the solution of this linear program can be reduced to 
other particular geometric set-cover problems.
We then develop algorithms for these sub-problems using the local ratio technique,
and Varadarajan's quasi-uniform sampling technique.

This general algorithmic approach improves the best known approximation ratios by at least an 
exponential factor (and much more in some cases) for essentially 
{\em all} of the nontrivial common special cases of this problem.
We believe that this geometric interpretation of scheduling 
is of independent interest.
\end{abstract}


\section{Introduction}

We consider the following general offline scheduling problem:

\medskip
{\em General Scheduling Problem (GSP):}
The input consists of a collection of $n$ jobs, and for each job $j$
a positive integer release time $r_j$,
a positive integer size $p_j$, 
and a cost or weight function $w_j(t) \geq 0$ for each $t > r_j$ 
(we are purposely not precise about how these weight functions are represented
in the input). 
Jobs are to be scheduled preemptively on
one processor after their release times.  If job $j$ completes at time $t$, then a
cost of $\sum_{s=r_j +1}^t w_j(t) $ is incurred. The scheduling objective is to minimize
the total cost, $\sum_{j=1}^n \sum_{s=r_j +1}^{C_j} w_j(t) $, where $C_j$ is the completion
time of job $j$.

This general problem generalizes several natural scheduling problems, for example:

\medskip
{\em Weighted Flow Time:}
If $w_j(t)  = w_j$, where $w_j$ is some fixed 
weight associated with job $j$, then the objective is weighted flow time.

\smallskip

{\em Flow Time Squared:}
If $w_j(t) = 2(t-r_j) -1$, then the objective is the sum of the squares of the flow times. 

\smallskip

{\em Weighted Tardiness:}
If $w_j(t) = 0$ for $t$ not greater than some deadline $d_j$,
and $w_j(t) = w_j$ for $t$ greater than $d_j$,
then the objective is weighted tardiness.
\medskip

In general, this problem formulation can model any cost objective function that is the 
sum of arbitrary cost functions for individual jobs, provided these cost functions are non-decreasing, i.e. it cannot hurt to finish a job earlier.

Flow time, which is the duration of time $C_j - r_j$ 
that a job is in the system,
is clearly the most natural and most commonly used quality of service measure for a job
in the computer systems literature.
Many commonly-used and commonly-studied scheduling objectives are based on combining
the flow times of the individual jobs.
However, flow time is also considered a rather difficult measure to work with mathematically. 
One reason for this is that even slight perturbations to the instance, can lead to
lead to large changes in the optimum value. 
Despite much interest, large gaps remain in our understanding for even basic 
flow time based scheduling objectives.
For example, for weighted flow time, the best known approximation ratios achievable by
polynomial-time algorithms are essentially no better than the poly-logarithmic 
competitive ratios achievable by online algorithms. 
For weighted tardiness, and flow time squared, no nontrivial approximation ratios were 
previously known
to be achievable.
While in contrast, for all of these three problems,
even the possibility of a polynomial time approximation scheme (PTAS) 
has not been ruled out.
We discuss the related previous work further in Section \ref{s:related}.


\subsection{Our Results}
The main contribution of this paper is the design and analysis of
a randomized $O(\log \log nP)$-approximation
algorithm for GSP,  where $P$ is the maximum job size. 
In the special case when all the release times are 0, we
obtain an  $O(1)$-approximation algorithm.
Let $W = \max_{j,t} w_j(t) $ be the maximum value attained by any weight function.
The running time of our algorithm is polynomial in $n$, $\log P$ and $\log W$,
provided that we can in polynomial time determine the times when a weight
function doubles. 
This is polynomial in the input size if the input 
must contain an explicit representation of the largest possible weight.


The primary insight to obtain these results is to view the scheduling problem geometrically.
The initial step is to show that GSP can be
reduced 
(with only a constant factor loss in the approximation ratio)
to the following geometric set-cover problem that we call R2C:

\bigskip

{\em Definition of the R2C Problem:}
The input consists of a collection of ${\mathcal P}$ points in two dimensional space,
and for each point $p \in {\mathcal P}$ an associated positive integer demand $d_p$.
Each point $p\in {\mathcal P}$ is specified by its coordinates $(x_p,y_p)$.
Further the input contains a collection ${\mathcal R}$ of axis-parallel rectangles, each of them abutting 
the $y$-axis. That is, each rectangle $r\in {\mathcal R}$ has the form  $(0,x_r) \times (y_r^1,y_r^2)$. 
In addition, each rectangle $r \in {\mathcal R}$ has an associated positive integer capacity $c_r$ 
and positive integer weight $w_r$.
The goal is to find a minimum weight subset $S \subset {\mathcal R}$ of rectangles, 
such that for each point $p \in {\mathcal P}$, the total capacity of rectangles covering $p$ is at least $d_p$,
that is, $ \sum_{r \in {\mathcal R} : p \in {\mathcal R}} c_r \geq d_p$.

\bigskip
As we shall see later, job sizes will be mapped to rectangle capacities in our reduction,
so we will also use $P$ to denote the largest capacity of any rectangle.
Our algorithm for R2C starts with the natural linear programming (LP) 
relaxation of the problem, 
strengthened by adding the so-called knapsack cover inequalities.
To round this LP solution, our algorithm then proceeds in a way that is by now standard (see for example \cite{CGK10}) in the applications of knapsack cover inequalities. In the terminology of \cite{CGK10}, we reduce the problem to rounding an LP solution for the so-called {\em priority} set cover version of the problem and in addition
several set multi-cover problems. These resulting problems are simpler as they are uncapacitated.

In particular we proceed as follows.
The algorithm first picks rectangles that are selected by
the LP solution to a significant extent (i.e. $x_r \geq \beta$, for some fixed constant $\beta$),
and then considers the {\em residual} solution.
The knapsack cover inequalities guarantee that remaining LP variables for a feasible solution 
to the residual instance. Since all variables $x_r \leq \beta$ in this solution,
the capacities and demands can be rounded to powers of $2$, and the variables can be scaled by
a constant factor, so that each point's demand is covered several times over. 

Points are then classified as heavy or light depending on whether or not
the optimal LP solution extensively covers the point with rectangles whose capacity
is larger than the demand of the point. 
We reduce the problem of covering the heavy points by rectangles with higher capacity to the
geometric cover problem R3U defined below. 
We show that the instances of R3U that we obtain 
have boundaries with low union complexity.
In particular, the boundary of the union of any $k$ objects 
has a complexity of $O(k \log P)$.
Using Varadarajan's quasi-uniform sampling technique~\cite{Varadarajan09} 
for approximating weighted set cover on geometric instances with low union 
complexity, one can obtain a covering that is  an
$O(\log \log nP)$-approximation to fractional cover specified by the
LP solution.

\medskip

{\em Definition of the R3U Problem:}
The input consists of a collection of ${\mathcal P}$ points in three dimensional space.
Each point $p\in {\mathcal P}$ is specified by its coordinates $(x_p,y_p, z_p)$.
Further the input contains a collection ${\mathcal R}$ of axis-parallel right cuboids
each of them abutting the $xy$ and $yz$ coordinate planes. 
That is, each right cuboid $r\in {\mathcal R}$ has the form  
$(0,x_r) \times (y_r^1,y_r^2) \times (0, z_r)$. 
In addition, each right cuboid $r \in {\mathcal R}$ has an associated positive 
integer weight $w_r$.
The goal is to find a minimum weight subset $S \subset {\mathcal R}$ of cuboids
such that each point $p \in {\mathcal P}$ is covered by at least one cuboid.

\medskip

We reduce the problem of covering the light points to $\log P$
different instances, one for each possible job size, 
of the weighted geometric multi-cover problem R2M defined below. 
We then show how to use the local ratio technique to obtain a solution
for each instance of R2M that is $O(\log \log nP)$-approximate with the cost in the
optimal LP solution for jobs of this size. Combining these solutions for various sizes
implies a solution for covering all light points with cost $O(\log \log nP)$ times the LP cost.

\medskip

{\em Definition of the R2M Problem:}
The input consists of a collection of ${\mathcal P}$ points in two dimensional space,
and for each point $p \in {\mathcal P}$ an associated positive integer demand $d_p$.
Each point $p\in {\mathcal P}$ is specified by its coordinates $(x_p,y_p)$.
Further the input contains a collection ${\mathcal R}$ of axis-parallel rectangles, 
each of them abutting the $y$-axis. 
That is, each rectangle $r\in {\mathcal R}$ has the form  $(0,x_r) \times (y_r^1,y_r^2)$. 
In addition, each rectangle $r \in {\mathcal R}$ has an associated 
positive integer weight $w_r$.
The goal is to find a minimum weight subset $S \subset {\mathcal R}$ of rectangles, 
such that for each point $p \in {\mathcal P}$, the number of 
rectangles covering $p$ is at least $d_p$.

\subsection{ Identical Release Times}
In the instances of R2C that arise from our reduction from the 
general scheduling problem, in the special case of identical release times,
all the points lie on a line, and the rectangles are one-dimensional intervals.
This is precisely the generalized caching problem, for which a  polynomial-time
4-approximation algorithm is known~\cite{BBF} (see also~\cite{CGK10}, for a somewhat more systematic approach to it). 
Thus we conclude that there is a polynomial-time $O(1)$-approximation 
algorithm for GSP
when all release times are identical.

\subsection{ Related Results}
\label{s:related}
Let us first consider weighted flow time.
\cite{BansalD07}
gives an online algorithm that is 
$O(\log W)$-competitive, and a semi-online algorithm 
(which means that the parameters $P$ and $W$ must be known
a priori to the online algorithm) that 
is $O(\log nP)$-competitive. \cite{ChekuriKZ01} gives a semi-online algorithm that is $O(\log^2 P)$-competitive.
These online algorithms also give the best known approximation ratios 
for polynomial time algorithms.
\cite{ChekuriK02}
gives a $(1+\epsilon)$-approximation algorithm that 
 has running time $n^{O((\log P \log W)/\epsilon^3)}$. 
Thus, this gives a quasi-polynomial time approximation
scheme (QPTAS) when both $P$ and $W$ are polynomially bounded in $n$. 
Moreover, \cite{ChekuriK02} also gives a QPTAS for the case when only one of either $P$ or
$W$ is polynomially bounded in $n$. 
In the special case that the weights are the reciprocal of the job sizes, and hence the objective is
average stretch/slow-down, then there is a polynomial time approximation
scheme~\cite{BenderMR04,ChekuriK02}.

It is also known that the algorithm highest density first is 
$(1+\epsilon)$-speed $O(1)$-competitive for weighted flow~\cite{BecchettiLMP06}
and flow squared~\cite{BansalP03}. No other approximation guarantees
are known for flow squared.
An $n-1$-approximation algorithm is known for weighted tardiness if
all jobs are released at the same time~\cite{Cheng2005}, and nothing seems to be known for arbitrary release dates.
PTAS's are known with the additional restriction that there are only a constant
number of deadlines~\cite{KarakostasKW09} or if jobs  have unit 
size~\cite{Lawler1982}. In general, there has been other extensive work on flow time related objectives and we refer the reader to  \cite{PST} for a survey.

The goal in geometric set cover problems is to improve the $O(\log n)$ set-cover bound using 
geometric structure. This is an active area of research and various different techniques have been developed.
However, until recently most of these techniques applied only to the {\em unweighted} case. 
A key idea is the connection between set covers and  
$\epsilon$-nets  \cite{BronnimannG95}, where an $\epsilon$-net is a sub-collection of sets that covers all the points that lie in at least an $\epsilon$ fraction of the input sets. 
For any geometric problem, existence of 
 $\epsilon$-nets of size at most $(1/\epsilon) g(1/\epsilon)$ implies $O(g(OPT))$-approximate solution for unweighted set cover~\cite{BronnimannG95}.
Thus, proving better bounds on sizes of $\epsilon$-nets  (an active research of research
is discrete geometry) directly gives improved guarantees for unweighted set-cover.
In a surprising result,~\cite{ClarksonV07} related the guarantee for unweighted set-cover
to the union complexity of sets.
If particular, if the sets have union complexity $O(n h(n))$,
which roughly means that the number of points on the boundary of
the union of any collection of $n$ sets is $O(n h(n))$,
then one can obtain an $O(h(n))$ approximation~\cite{ClarksonV07}. This was subsequently improved 
to $O(\log (h(n))$ \cite{Varadarajan09}.
In certain cases these results also extend to the  unweighted multi-cover case \cite{CCS}.
However, these techniques do not apply to weighted set cover problems: the problem is that
these techniques may sample some sets with much higher 
probability than that specified by the LP relaxation. 
In a recent breakthrough,  Varadarajan gave a new quasi-uniform sampling technique~\cite{Varadarajan10}
that obtains a $ 2^{O(\log^* n)} \log(h(n))$ approximation for weighted geometric set cover problems
with union complexity $O(nh(n))$. In fact his result gives an improved guarantee of $O(\log h(n))$ if $h(n)$ grows with $n$ (even very mildly such as $\log \log \cdots \log n$, where the $\log$ is iterated  $O(1)$ times).

{\em Organization:}
The paper is organized as follows. In
section \ref{sec:reduce} the reduction from GSP to R2C is given.
In section \ref{sec:preliminaries} we give the LP formulation
of R2C and explain the initial preprocessing of the LP solution.
In section \ref{sec:heavy}
we explain how to reduce part of the problem
of rounding the LP solution to an instance of the R3U problem.
In section \ref{sec:light}
we explain how to reduce part of the problem
of rounding the LP solution to an instance of the R2M problem.

\section{The Reduction from GSP to R2C}
\label{sec:reduce} 
Our goal in this section is to prove Theorem \ref{red}.
We accomplish this by giving a reduction from GSP
to R2C, and then showing that this reduction increases the objective 
value of the optimal solution by at most a factor of four
(Lemma \ref{gc:s}), and that this reduction doesn't shrink the
objective value of the optimal solution (Lemma \ref{s:gc}).

\begin{thm}
\label{red}
A polynomial-time $\alpha$-approximation algorithm for R2C 
implies a polynomial-time $4\alpha$ approximation algorithm for GSP.
\end{thm}

\medskip

{\em Definition of the Reduction from GSP to R2C:}
From an arbitrary instance $\mathcal I$ of GSP,
we explain how to create an instance $\mathcal I'$ of R2C.
Considering $\mathcal I$,
we say that a time $t > r_j$ is of class $k \geq 1$ with respect to job $j$ if 
the cost of finishing $j$ at time $t$ lies in $[2^{k-1},2^k-1]$, 
i.e.  $\sum_{t'=1}^t w_{j}(t')  \in [2^{k-1},2^k-1]$. We say that $t$ is of class $0$, if the cost of finishing $j$ at $t$ is $0$.
Let $I_k^j$ denote the (possibly empty) time interval of class $k$ times with respect to job $j$.
Let $\cal{T}$ denote the set of all points that are endpoints of the intervals 
of the form $I_k^j$ for some job $j$ and class $k$. 
For each time interval $X$ of the form $X=[t_1,t_2)$, 
where $t_1 < t_2$ and $t_1,t_2 \in \cal{T}$, we create a point $p$ in $\mathcal I'$
with demand
$d_p = \max(0,P(X) - |X|) = \max(0,P(X)-(t_2 - t_1))$, where $P(X)$ denotes the total size of jobs that are released during $X$, 
i.e. $P(X) = \sum_{j: r_j \in [t_1,t_2)} p_j$.
For each job $j$ in $\mathcal I$ 
and $k\geq 0$, we create a rectangle $R^j_k = [0,r_j] \times I^j_k$
in $\mathcal I'$ with capacity $p_j$ and weight $2^k-1$.
We note that the rectangles $R^j_0,R^j_1,\ldots$ corresponding to the same job are pairwise disjoint.

Without loss of generality, we may assume that the time horizon is $nP$,
otherwise the instance can be divided into disjoint non-interacting subsets.
Thus the maximum cost for any job can be $nPW$, so $k \leq \min(nP,\log(nPW)$. This implies that 
we can assume that $\log W = O(nP)$ and that  $|\mathcal{T}| = O(n \log (nPW) )$, 
i.e. polynomial in the size of the input. 
Throughout the paper we will use $m$ to denote the number of points in the R2C problem.
Clearly, $m = O(|\mathcal{T}|^2)$. 

\begin{lemma}
\label{gc:s}
If there is a feasible solution $S$ to $\mathcal I$ 
with objective value $v$, 
then there there is a feasible solution $S'$ to $\mathcal I'$
with objective value at most $4v$.
\end{lemma}
\begin{proof}
For job $j$ in $\mathcal I$, 
let $k(j)$ denote the class during which $j$ finishes in $S$
(i.e. $k(j)$ is the smallest integer such that the cost 
incurred by $j$ in $S$ is $\leq 2^{k(j)}-1$).
Consider the solution $S'$ obtained by choosing for each job $j$, 
the intervals $I^j_0, \ldots, I^j_{k(j)}$.
Clearly, each job contributes at most $\sum_{i=0}^{k(j)} 2^i-1 \leq 2(2^{k(j)}-1) \leq 4 \cdot 2^{k(j)-1}$,  i.e. at most 4 times its contribution to $S$, and hence the total cost of $S'$  
is at most $4$ times the cost of $S$.

It remains to show that  $S'$ is feasible, i.e. for any point $p$, the total capacity of rectangles covering $p$ is at least $d_p$. 
Suppose $p$ corresponds to the time interval $X= [t_1,t_2)$ from $\mathcal I$. 
Let $J_X$ denote the jobs that arrive during $X$. 
For  each job $j \in J_X$ that completes after $t_2$, 
there is exactly one rectangle $R^j_k$ that covers $p$. 
Since $S$ is a feasible schedule, the total size of jobs in $J_X$ 
that can complete during $X$ itself cannot be more than 
$|X|=t_2-t_1$. Thus the jobs in $J_X$ that do not complete during $X$ 
must have a total size of at least $P(J_X) - |X|$,
which is the covering requirement for $p$.
\end{proof}

\begin{lemma}
\label{s:gc}
If there is a feasible solution $S'$ to $\mathcal I'$ 
with objective value $v'$, 
then there there is a feasible solution $S$ to $\mathcal I$
with objective value at most $v$.
\end{lemma}
\begin{proof}
For each job $j$, let $h(j)$ denote the largest index such that the rectangle $R^j_{h(j)}$ lies in $S'$.
Let us set a deadline $d_j$ for $j$ as the right end point of $I^j_{h(j)}$.

We claim that there is a schedule $S$ that completes each job
$j$ by time $d_j$. Consider the bipartite graph defined as follows: We have time slots $1,2,\ldots,T$ on the right. For each job $j$, we have
$p_j$ vertices on the left, each of which is connected to vertices $r_j,\ldots,d_j-1$ on the right.
By Hall's theorem, a feasible schedule exists if and only if for any time interval $X$, the total size of jobs
that have both release times and deadlines in $X$ is at most $|X|$.
Moreover, it suffices to show such a result for intervals $X$ of the form 
$[r_a,d_b)$, for some jobs $a$ and $b$.
Equivalently, for any such time interval $X$, 
the jobs $j \in J_X$ that are released during $X$ and have $d_j$ after the end of $X$,
have a total size of at least $P(J_X) - |X|$. 
 
Note that by the definition of $\mathcal T$, then there is a point
$p$ in $\mathcal I'$ that corresponds to the interval $X$.
Then by the feasibility of $S'$, 
the total capacity of rectangles covering $p$ in $S'$
is at least $P(J_X) - |X|$. And as all of these rectangles 
correspond to different jobs in $\mathcal I$ (the rectangles corresponding to the same job are pairwise disjoint), we are done.

In $S$ the cost of $j$ is at most $2^{h(j)}-1$, since by the definition of the rectangle $R^j_k$ the cost of finishing a job by deadline $d_j$ is at most $2^{h(j)}-1$. Now, the cost incurred by $j$ in $I'$ is at least $2^{h(j)}-1$ (since the rectangle $R^j_{h(j)}$ already has cost $2^{h(j)}-1$). This implies that the cost of $S$ is at most that of $S'$.
\end{proof}

\bigskip

{\em Identical Release times:}
Without loss generality, let $r_j=0$ for all $j$. In this case, the above reduction become simpler. In particular, the first dimension corresponding to release time becomes irrelevant and we obtain the following problem. For each job $j$ and $k\geq 0$, there is an interval $I^j_k$ corresponding to class $k$ times with respect to $j$ and has capacity $p_j$ and weight $2^{k}-1$.
All relevant intervals $X$ are of the form 
$[0,t]$ for $t \in \cal{T}$ and have demand $J_X-|X|=D-t$, where $D$ is the total size of all the jobs. 
For each such $X=[0,t)$, we introduce a point $t$ with demand $d_t=D-t$.
The goal is to find a minimum weight subcollection of intervals $I^j_k$ such that covers the demand.
This is a special case of the following Generalized Caching Problem.

\medskip
{\em Generalized Caching Problem:} The input consists of a set of demands $d(t)$ at various time steps $t=1,\ldots,n$. In addition there is a collection of 
time intervals  $\cal{I}$, where each interval $ I \in \mathcal{I}$ has weight $w_I$, size $c_I$ and span
$[s_I,t_I]$ with $s_I,t_I \in \{1,\ldots,n\}$. The goal is to find a minimum weight subset of intervals that covers the demand. That is, find the minimum weight subset of intervals $S \subseteq \mathcal{I}$ such that 
$$ \sum_{I \in S: t \in [s_I,t_I]} c_I \geq d_t \qquad \forall t \in \{1,\ldots,n\}.$$

A $4$-approximation for this problem was obtained by Bar-Noy et al. \cite{BBF}, based
on the local-ratio technique. Their algorithm can equivalently be viewed as a primal dual algorithm
applied to a linear program with knapsack cover inequalities \cite{BR}.
This immediately implies a $16$-approximation for GSP in the case of identical release times.

\section{The LP Formulation for R2C}
\label{sec:preliminaries}

The following is a natural integer programming formulation for
R2C. For each rectangle $r \in {\mathcal R}$ there is an indicator variable
$x_r$ specifying whether or not the rectangle $r$ is selected.
\begin{eqnarray}
\min\sum_{r \in {\mathcal R}} w_r x_r &&  \qquad \mbox{s.t.} \nonumber \\
\label{cap:simple} 
\sum_{ r : p \in r}  c_r x_r &\ge& d_p  \qquad \qquad \forall p \in {\mathcal P} \\
 x_r & \in &  \{0,1\} \qquad \qquad r \in {\mathcal R} 
\end{eqnarray}

It is easily seen 
that the natural relaxation of this linear program, where $x_{r} \in \{0,1\}$ 
is replaced by $x_{r} \in [0, 1]$, has a
large integrality gap. In particular, this is true even when ${\mathcal P}$ consist of a single point, in which case the problem is equivalent to the knapsack cover problem \cite{CarrFLP00}.
Thus, we strengthen this LP by adding knapsack cover inequalities introduced in \cite{CarrFLP00}
have proved 
to be a useful tool to address capacitated covering problems~\cite{gencaching, shmoys, sviri, BansalGK2010, CGK10}. 

This gives the  the following linear program:
\begin{align}
\label{a} \min\sum_{r \in {\mathcal R}}  w_{r} x_{r} && \mbox{s.t.}\\
\label{c}\sum_{ r \in {\mathcal R}\setminus S: p \in r} 
\nonumber \min \left\{ c_r,  \max(0,d_p - c(S)) \right\} x_{r} &\ge& \\
d_p - c(S) \qquad \forall p\in {\mathcal P}, S \subseteq {\mathcal R} &&\\
\label{b} x_r  \in  [0,1]  \qquad \forall r \in {\mathcal R} && 
\end{align}
Here $c(S)$ denotes the total capacity of rectangles in $S$.
The constraints are valid for the following reason: For any subset $S$,  even if all the items in $S$ are chosen, at least a demand of $d_p -c(S)$ must be covered by remaining rectangles. 
Moreover, truncating an item size to the residual capacity does not affect the feasibility of an integral solution. Even though there are exponentially many constraints per point, 
a feasible $(1+\epsilon)$-approximate solution, for any constant $\epsilon>0$, can be found using the Ellipsoid algorithm, see
\cite{CarrFLP00} for details. Further 
only the cost incurs the $(1+\epsilon)$ factor loss, 
all the constraints are satisfied exactly.
We will refer the inequalities in line (\ref{c}) 
as the {\em knapsack cover inequalities}.

Let $x$ be some $(1+\epsilon)$-approximate feasible solution to the linear
program for R2C in lines (\ref{a})-(\ref{b}), 
and let $\OPT$ denote $x$'s objective value.

We now apply some relatively standard steps to simplify $x$.
Let $\beta$ be  a small constant, $ \beta = 1/12$ suffices. 
Let $S$ denote the set of rectangles for which $x_r\geq \beta$. 
We pick all the rectangles in $S$, i.e. set $x_r=1$. Clearly, this cost of this set is at most $1/\beta$ times the LP solution.

For each point $p$, let $S_p = S \cap \{r: r \in {\mathcal R},p \in r\}$ denote the set of rectangles in $S$ that cover $p$. Let us consider the residual instance, where the set of rectangles is restricted to ${\mathcal R} \setminus S$ 
and the demand of a point is $d_p - c(S_p)$. If $d_p-c(S_p) \leq 0$, then $p$ is already covered by $S$ and we discard it.

Since the solution $x$ satisfied all the knapsack cover inequalities for each point $p$ and set $S$, and hence in particular 
for every $p$ and corresponding the set $S_p$, we have
that 
$$ \sum_{ r \in {\mathcal R}\setminus S_p: p \in r} 
\min \left\{ c_r, d_p - c(S_p) \right\} x_{r} \ge d_p - c(S_p) $$
Henceforth, this is the only fact we will use about the solution $x$ (in particular, we do not care that $x$
satisfies  several other inequalities for each point $p$).
Let us scale the solution $x$ restricted to ${\mathcal R} \setminus S$ by $1/\beta$ times. Call this solution $x'$. Note that since $x_r \leq \beta$, it still holds that   $x'_r \in [0,1]$.  Clearly, $x'$ satisfies
$$ \sum_{ r \in {\mathcal R}\setminus S_p: p \in r} 
\min \{ c_r, d_p - c(S_p) \} x'_{r} \ge \frac{d_p - c(S_p)}{\beta} $$
Let us define the new demand $d'_p$ of $p$ as 
$d_p - c(S_p)$ rounded up to the nearest integer power of $2$.
Similarly, defined a new capacity $c'_r$ of 
each rectangle $r$ to be $c_r$ rounded down to the nearest integer power of $2$.
$x'$ still satisfies,
$$ \sum_{ r \in {\mathcal R}\setminus S_p: p \in r}  \min \{ c'_r, d'_p \} x'_{r} \ge  \frac{d'_p}{4\beta} $$

We call $r$ a class $i$ rectangle if $c'_r=2^i$. Similarly, $p$ is a class $i$ point if $d'_p=2^i$.
We call a point $p$ {\em heavy} if is covered by rectangles with class at least as high as that of $p$ in the LP solution, more precisely
if:
\begin{equation}
\label{eq:heavy}
\sum_{r \in {\mathcal R}': c'_r \geq  d'_p} \min(c'_r, d'_p) x'_{r} \ge d'_p. 
\end{equation}
Equivalently, $p$ is heavy if 
$$\sum_{r \in {\mathcal R}': c'_r \geq  d'_p}  x'_{r} \ge 1.$$ 
Otherwise we say that a point is {\em light}. Thus a light point satisfies:
\begin{equation}
\label{eq:med1}
\sum_{r \in {\mathcal R}': c'_r \leq  d'_p}  c'_r x'_{r} \ge 
\left(\frac{1}{4\beta}-1\right) d'_p =
\left(\frac{1 - 4 \beta}{4\beta}\right) d'_p
\end{equation}
We now have different algorithms for covering heavy and light points.

\section{Covering Heavy Points}
\label{sec:heavy}

In this section 
we show how reduce the problem of covering the heavy
points by larger class rectangles to R3U.
We then show that the resulting instances of R3U have low union complexity.
In particular any $k$ cuboids in a resulting R3U instance
has union complexity $O(k \log P)$.
By Varadarajan's quasi-uniform sampling technique~\cite{Varadarajan09}
this gives 
a solution that is an $2^{O(\log^*m)} \log \log P = O(\log \log nP)$ approximation to the 
optimal fractional solution of this R3U instance. 
As $x'$ gives
a feasible fractional solution to this R3U instance, this means
that the cost of cuboids that the algorithm selects is  $O(\log \log nP)$
approximate with $\OPT$.

\medskip

{\em The Problem of 
Covering the Heavy Points to R3U:} 
The reduction takes as inputs 
the instance $\mathcal I'$ for heavy points obtained at the end of the previous section, and the LP solution $x'$ 
and creates an instance $A$ of R3U.
For each heavy point $p = (x, y) \in {\mathcal I'}$ with demand $d'_p$, 
there is a point $(x,y, d'_p)$ in $\mathcal A$.
For each rectangle $r = [0, x] \times [y_1, y_2]$ in $\mathcal I'$ with capacity $c'_r$, 
we define a right cuboid $R_r = [0,x]\times [y_1, y_2] \times [0,c'_r]$
of weight $w_r$.

\medskip
It is clear that there is a one to one correspondence between
a covering of heavy points in $\cal I'$ by rectangles of no smaller
class and a covering of the points in $A$ by cuboids. 
Given a collection $X$ of $n$ geometric objects, the union complexity of $X$ is
number of edges in the arrangement of the boundary of $X$. For 3-dimensional objects, this is the 
total number of vertices, edges and faces on the boundary of $X$. 
In Lemma \ref{2dunion} and Lemma \ref{3dunion}
we bound the union complexity of cuboids in $A$.

\begin{lemma}
\label{2dunion}
For any collection of $k$ rectangles of the type  $[0,r]\times [s,t]$, the union complexity  is $O(k)$.
\end{lemma}
\begin{proof}
For each rectangle of the form $[0,r]\times[s,t]$ has a side touching the $y$-axis. Let us view of union of $k$ such rectangles from $(\infty,0)$. Consider the vertical faces on the boundary of the union.
For any two rectangles $a$ and $b$, the pattern $abab$ or $baba$ cannot appear. Thus the vertical faces from a Davenport Schinzel sequence of order 2, which has size at most $2k-1$ (see for example \cite{Mat}, chapter 7). Since the number of vertices is $O(1)$ times the number of faces, the result follows. 
\end{proof}

\begin{lemma}
\label{3dunion}
The union complexity of any $k$ cuboids in $\mathcal{R}$ is $O(k \log P)$.
\end{lemma}
\begin{proof}
This directly follows from lemma \ref{2dunion} and noting that the number of distinct heights is $O(\log P)$. In particular, since the heights of powers of 2, consider the  slice of the arrangement between $z=2^{i}$ and $z=2^{i+1}$. This corresponds to union of rectangles of the form $[0,r]\times [s,t]$. 
\end{proof}

\noindent
{\em Remark:} We remark that the bound in lemma \ref{3dunion} is tight for kind of cuboids we consider here.

The following result  is implicit in \cite{Varadarajan10}.
\begin{thm} [\cite{Varadarajan10}]
There is a randomized polynomial-time algorithm that,
given a weighted geometric set cover instance $I$
where the union complexity of any $k$ objects is $k* g(k)$,
produces an set cover of weight at most a factor of $2^{O(\log^* |I|)} \log g(|I|)$ times the 
optimal fractional set cover.

If the function $g(n)$ grows even very mildly with $n$, say in particular that $g(n) \geq \log \log \cdots \log n$, where the $\log$ is iterated $O(1)$ times, then the approximation guarantee above is $O(\log g(|I|))$.
\end{thm}

Thus we can conclude that in polynomial time one find rectangles
in the R2C instance $\cal I'$ that 
covers all the heavy points and that has weight at most 
$O(\log \log nP )$ times $\OPT$.

\section{Covering Light Points}
\label{sec:light}

In this section 
we show how to decompose the problem of covering the light 
points to $\log P$ instances of R2M, one instance $B_\ell$ for 
each possible rectangle capacity class $\ell$. The decomposition ensures that an $\alpha$ approximation for R2M implies an cover for light points in $I'$ with cost 
$O(\alpha)$ times $\OPT$.
We then give an obtain an $O(\log \log m) = O(\log \log nP)$ approximation for an R2M instance on $m$ points. 
To do this, we relate the multi-cover 
problem to the set cover problem (where all demands are $1$) and show that the set cover problem has
a 2-approximation with respect to the fractional solution. 
This implies that the cost of rectangles that the algorithm selects for $\cal I'$ is $O(\log \log m)$
approximate with $\OPT$.

\noindent
{\em Remark:} Better results for the R2M problem can be obtained by adapting Varadarajan's quasi-uniform
sampling technique to multi-cover instances. However, we follow the simpler approach here since it suffices for our purposes.

\medskip

{\em The Problem of 
Covering the Light Points to the instances $B_\ell$ of R2M:} 
The reduction takes as inputs 
the instance $\mathcal I'$ for R2C (restricted to light points), and the LP solution $x'$ 
and for each $\ell =0,1,2,\ldots$ creates an instance $B_{\ell}$ of R2M.
The points in $B_\ell$ are the same as the points in $\mathcal I'$.
The demand of a point $p$ in $B_\ell$ is defined as 
$d^{\ell}_{p}=\lfloor \sum_{r : c'(r) = 2^\ell} x_r' \rfloor$.
The rectangles in $B_\ell$ are precisely the class $\ell$ rectangles in $\cal I'$, i.e. those of capacity  exactly $2^\ell$.
The weight of the rectangles in $B_\ell$ are the the same as in $\cal I'$.
The goal is to cover each point $p \in B_{\ell}$ by $d^{\ell}_p$ distinct rectangles. 


\begin{lemma}
Consider the union $S$ of the rectangles picked in the solutions $S_\ell$
to the instances ${\mathcal B}_\ell$.
Then $S$ satisfies the demand of  all the light points in $\mathcal I'$.
\end{lemma}
\begin{proof} 
Consider a particular point $p$ and suppose it lies in class $i$ in $\cal I'$, i.e. its demand $d'(p)=2^i$. Then the extent to which $p$ is
covered by $\bigcup_\ell S_\ell$ is at least
{\allowdisplaybreaks
\begin{eqnarray*}
\sum_{\ell < i} 2^\ell d^{\ell}_p  & =  &\sum_{\ell < i} 2^\ell \lfloor \sum_{r : c'(r) = 2^\ell \mbox{ and } p \in r}  x_r' \rfloor \\
&\ge&
\sum_{\ell < i} 2^\ell ( (\sum_{r : c'(r) = 2^\ell \mbox{ and } p \in r}  x_r')-1)  \\
&\ge& \left(\sum_{\ell < i} 2^\ell  \sum_{r : c'(r) = 2^\ell \mbox{ and } p \in r}  x_r'  \right) - 2^i \\
&= & \left( \sum_{\ell < i} 2^\ell  \sum_{r : c'(r) = 2^\ell \mbox{ and } p \in r}  x_r' \right)  - d'(p) \\
&\ge & \left(\frac{1 - 8 \beta}{4\beta} \right)  d'(p) 
\end{eqnarray*}
}
where last inequality follows from (\ref{eq:med1}).
Since $\beta=1/12$, it follows the each $p$ is covered.
\end{proof} 

Henceforth we focus on a particular instance of R2M.
Let $I$ be such an instance with $n$ rectangles (sets) $S_1,\ldots,S_n$ and $m$ points (elements) $1,\ldots,m$. Let $d_i$ denote the covering requirement of $i$. 
We are given  some fractional feasible solution $x$, i.e. for each $i$  $\sum_{j: i \in S_j } x_j \geq d_i$ and $x_j \in [0,1]$ for all $S_j$.
The following lemma is standard.
\begin{lemma}
\label{lowdem}
For any multi-cover problem, 
at the loss of an $O(1)$ factor in approximation ratio, we can assume that the maximum demand $d=\max_i d_i$ is $O(\log m)$. 
\end{lemma}
\begin{proof}
We pick each set $S_j$ with probability $\min(1,2x_j)$. The expected cost of the sets picked is at most twice the LP cost.
By standard Chernoff bounds, for some large enough constant $c$
each element with demand $d_i \geq c \log m$ is covered 
with probability at least $1-1/m^2$. In the residual instance, each uncovered element has demand $O(\log m)$ 
and as $x_j \leq 1$ for each set, the LP solution restricted to the unpicked sets is a feasible solution to the residual instance.
\end{proof}

The following lemma shows how a rounding procedure for a set cover problem can be used for corresponding  multi-cover problem.
 
\begin{lemma}
\label{setmult}
An LP-based $\alpha$ approximation algorithm for a weighted set cover problem can be used to obtain an $\alpha \log d$ approximation for any multi-cover variant of the problem where $d$ is the maximum demand of any element. 
\end{lemma}
\begin{proof}
Let $x$ be some feasible fractional solution to the multi-cover problem. Our algorithm proceeds in $d$ rounds, and picking 
some sets in each round such that after $d$ rounds, each $p_i$ is covered by at least $d_i$ distinct sets. 
Inductively, assume
that at beginning of round $r$ each element has an uncovered demand of at most $d-r+1$.
This is clearly true for $r = 1$.
For round $r=1,\ldots,d$, we proceed as follows. 
Consider the LP solution $y^{(r)} =x/(d - r + 1)$, restricted to the sets not chosen thus far in previous rounds.
Let $P_r$ be the elements with (current) demand exactly $d-r+1$.
We claim that $y^{(r)}$ is a feasible fractional set cover solution for $P_r$.
If $i \in P_r$ had requirement $d_i$ initially, then it has been covered $c_i=d_i - (d-r+1))$ times thus far. As each $x_j \leq 1$, the solution $x$ restricted to sets not picked this far still covers $i$ to extent $d_i -c_i$ and hence $y^{(r)}$ must cover $i$ fractionally to extent at least $(d_i-c_i)/(d-r+1) \geq 1$. 

Let $C_r$ denote the cover for $P_r$ obtained by applying our set cover rounding procedure to $y^{(r)}$.
We return the solution $C_1 \cup \ldots \cup C_d$. In this solution, each element $i$ is covered at least $d_i$ times, and its cost is $ \sum_{r=1}^d \alpha \cdot \textrm{cost} (y^{(r)}) \leq \sum_{r=1}^d \alpha \cdot \textrm{cost} (x/(d-r+1)) = \alpha \log d \cdot \textrm {cost} (x)$. 
\end{proof}

We now give a $2$ approximation for R2M using local ratio. We refer the reader to \cite{BBF} for a general description of the technique. While we use local ratio below, our approximation can be easily made LP-based
using the equivalence between local ratio and the primal dual method \cite{BR}.

\begin{lemma}
\label{r2mlemma}
There is a 2-approximation for the R2M problem when all the demands are $O(1)$.
\end{lemma}
\begin{proof}
The algorithm is a straight-forward
application of local ratio rule. We adopt the notation from all the local ratio rule papers. Let $w$ be
the original weight function. Consider the rightmost point $p$ to be covered, that is the point $p$ with
maximum $x$ coordinate (if there are several, pick one arbitrarily). 
Let $z$ be the minimum weight of a rectangle
covering $p$. Define the weight function $w_1 = z$ for rectangles that cover $p$, and 0 for the other rectangles. Let $w_2 = w - w_1$ be the residual weight function.
Recall that the local ratio rule tentatively picks all the sets $X$ with $w_2$ weight 0, removes the covered points and proceeds recursively on the residual instance with function $w_2$.
Let $S_2$ be the solution obtained recursively by the local ratio for the residual instance.
We then add all the rectangles in $X$ and perform the greedy-delete step, i.e. remove them arbitrarily as
long as solution is feasible.

As $p$ must be covered, any optimum solution must incur a $w_1$ cost of $z$.
It suffices to show that at most two rectangles with non-zero $w_1$ weight can be picked by the algorithm.
Suppose more than two are left after the delete step. But as $p$ is the rightmost point, any rectangle that covers $p$ and is different from the one with the topmost edge or the one with the  bottommost edge will be redundant.
\end{proof}

\section*{ Acknowledgments} We greatly thank Alexander Souza, Cliff Stein, Lap-Kei Lee, Ho-Leung Chan,
and Pan Jiangwei for extensive discussions about this research.

\bibliographystyle{plain}
\bibliography{focs-appear}

\end{document}